\documentclass[subscriptcorrection,upint,varvw,barcolor=black,mathalfa=cal=euler,balance,hyphenate,french,pdf-a,nolists]{asmejour} %

\JourName{ASME Letters in Dynamic Systems and Control}%<=== change to the name of your journal

\usepackage{amsthm}

\hypersetup{colorlinks=true,linkcolor=blue,citecolor=blue}
\graphicspath{{Figs/}}

\newtheorem{theorem}{Theorem}

\newtheorem{definition}{Definition}
\newtheorem{remark}{Remark}
\newtheorem{corollary}{Corollary}

\catcode`\@=11
\def\downparenfill{$\m@th\braceld\leaders\vrule\hfill\bracerd$}
\def\overparen#1{\mathop{\vbox{\ialign{##\crcr\crcr \noalign{\kern0.4ex}
\downparenfill\crcr\noalign{\kern0.4ex\nointerlineskip}
$\hfil\displaystyle{#1}\hfil$\crcr}}}\limits}
\catcode`\@=12

\usepackage{color}

\begin{document}

\SetAuthorBlock{Tianyu Han}{
   Department of Mechanical Engineering,\\
   The City College of New York,\\
   The City University of New York,\\
   160 Convent Avenue,\\
   New York, NY 10031 USA \\
   email: than000@citymail.cuny.edu} 

\SetAuthorBlock{Bo Wang\CorrespondingAuthor}{
   Department of Mechanical Engineering,\\
   The City College of New York,\\
   The City University of New York,\\
   160 Convent Avenue,\\
   New York, NY 10031 USA \\
   email: bwang1@ccny.cuny.edu} 

\title{Scaling-Based Reciprocal Control Barrier Functions for Nonholonomic Mobile Robots}


\begin{abstract}
This paper studies the construction of control barrier functions (CBFs) for force-controlled nonholonomic mobile robots subject to relative-degree-two safety constraints arising from position-level obstacle avoidance. A scaling-based reciprocal barrier construction is proposed, in which a positive motion-dependent scaling factor is placed in the numerator of a reciprocal barrier associated with the original physical safety function. The resulting barrier is defined exactly on the interior of the physical safe set and becomes singular on its boundary, thereby preserving the certified interior domain of the original safety constraint while recovering first-order control authority. For a force-controlled nonholonomic robot model, sufficient conditions are derived under which the proposed construction defines a reciprocal CBF, and the interior of the physical safe set is forward invariant under controllers satisfying the induced reciprocal-CBF condition. A scalar strict-feedback system is further used to provide a structural interpretation of the underlying higher-relative-degree cascade under explicit structural assumptions. Numerical simulations demonstrate the induced safe-set geometry and its integration with an optimization-based control framework for obstacle avoidance.

\end{abstract}

\date{}

\maketitle 

%% -------------------------------------------------------------------
\section{Introduction}

Control barrier functions (CBFs) have become a central tool in safety-critical control for enforcing forward invariance of prescribed safe sets \citep{ames2016control,Ames2019Control}. Their ability to encode safety conditions through inequality constraints has led to broad use in robotics and autonomous systems, including mobile robots \citep{Han2024Safety,Wang2025Further}, manipulators \citep{Singletary2022Safety}, and multi-agent systems \citep{jankovic2023multiagent}. Despite this success, the construction of effective barrier functions remains challenging for mechanical and robotic systems in which safety specifications are imposed at the configuration level while the available control inputs act through higher-order dynamics \citep{Cohen2024Safety}. In such cases, the physical safety function may have relative degree greater than one, so a standard first-order CBF condition cannot be imposed directly. The difficulty is particularly pronounced when one seeks a barrier construction that restores first-order control authority without replacing the original configuration-level safe set by a smaller certified set.

This difficulty arises naturally in obstacle avoidance for force-controlled nonholonomic mobile robots \citep{wang2022time,wang2022robust}. In such systems, collision avoidance is specified through a position-level safety constraint, whereas the control input acts at the acceleration level. Consequently, the corresponding physical safety function has relative degree two with respect to the input, and the standard first-order CBF condition cannot be applied directly. Existing approaches typically recover enforceability through recursive differentiation, additive augmentation, or activation mechanisms. While effective, these constructions may introduce auxiliary conditions or modified barrier domains that certify only a subset of the original physical safe set in some designs. These observations motivate the present work, which develops a scaling-based reciprocal barrier construction aimed at recovering first-order control authority while preserving the certified interior domain of the original physical safe set.

To address high-relative-degree safety constraints, existing CBF constructions typically modify the barrier condition so that control authority is recovered when the input does not directly affect the original safety function. High-order control barrier functions (HOCBFs) achieve this by recursively differentiating the safety function and imposing inequalities on successive derivatives until the control input appears~\citep{Xiao2022HOCBF}. This framework is general and widely used, but the resulting recursive conditions can introduce intermediate feasibility requirements that may restrict the certified set. Backstepping-based CBF constructions~\citep{Taylor2022Safe,Cohen2024Safety,Han2024Safety} recover control authority through virtual controls and tracking-error-type coordinates, providing a constructive approach for combining safety and stabilization; however, the induced barrier conditions may certify a domain that differs from the original physical constraint. Activation-based constructions, including rectified CBFs (ReCBFs)~\citep{Ong2024Rectified} and activated backstepping methods~\citep{Gacsi2025activedback}, introduce additional mechanisms to modify or activate barrier conditions when the nominal construction is ineffective, improving flexibility at the expense of extra design choices. Related output-based constructions~\citep{Cohen2024Constructive} define safety on suitable lower-dimensional outputs and lift the resulting conditions to the full dynamics, but their applicability depends on the availability of an appropriate output structure.

Despite these advances, a recurring difficulty remains in constructing CBFs that recover first-order enforceability while keeping the certified domain tied directly to the original physical safety constraint. This issue is especially relevant for position-level safety constraints in robotic systems, including force-controlled nonholonomic mobile robots. Rather than introducing additive, recursive, or auxiliary barrier terms that may alter the certified domain, this paper develops a scaling-based reciprocal barrier construction in which a positive state-dependent scaling factor is placed in the numerator of a reciprocal quotient associated with the physical safety function. The resulting barrier remains defined exactly on the interior of the original physical safe set and becomes singular as the state approaches the physical safety boundary. In this sense, preserving the original safe-set geometry means preserving the certified interior domain, while recovering explicit first-order control authority. 
The practical benefit targeted by this construction is preservation of the certificate domain, rather than a uniform reduction in closed-loop conservatism. This property is most relevant when maintaining a safety certificate over the entire geometrically safe domain is itself an important design objective. When bounded-input feasibility, robustness, or control-effort regulation is the primary concern, HOCBF- or backstepping-based constructions may instead be more appropriate. Thus, the proposed approach represents a different design trade-off, prioritizing certificate-domain preservation while treating feasibility and performance considerations separately.

\textit{Contributions.} This paper makes three main contributions. First, we propose a scaling-based reciprocal CBF construction for force-controlled nonholonomic mobile robots subject to relative-degree-two safety constraints arising from position-level obstacle avoidance. The construction places a positive, bounded, motion-dependent scaling factor in the numerator of a reciprocal barrier associated with the physical safety function. This design preserves the certified interior domain of the original physical safe set while restoring direct input dependence in the first-order barrier dynamics. Second, for the considered robot model, we prove that the resulting barrier is a reciprocal CBF and show that the interior of the physical safe set is forward invariant under controllers satisfying the induced reciprocal-CBF condition. 
Third, using a scalar strict-feedback system as a canonical cascade model, we provide a structural interpretation of the mechanism underlying the proposed construction under explicit relative-degree and control-authority assumptions.
Numerical simulations further show the induced safe-set geometry and its integration with an optimization-based obstacle-avoidance controller.

The remainder of this paper is organized as follows. Section~\ref{sec:formulation} reviews CBF preliminaries, introduces the force-controlled nonholonomic robot model, and formulates the problem. Section~\ref{sec:main} presents the proposed reciprocal CBF construction and the main invariance result. Section~\ref{sec:strict-feedback-system} provides a structural interpretation of the proposed mechanism via a scalar strict-feedback system.
Section~\ref{sec:simulations} provides numerical examples demonstrating the induced safe-set geometry and its integration with an optimization-based controller. Section~\ref{sec:conclusion} concludes the paper.

%% ----------------------------------------------------------------
\section{Preliminaries and Problem Formulation}\label{sec:formulation}

\textit{Notation.} The Euclidean norm on $\mathbb{R}^n$ is denoted by $|\cdot|$. For a set $S\subset\mathbb{R}^n$, $\partial S$ and $\operatorname{Int}(S)$ denote its boundary and interior, respectively. A continuous function $\alpha:\mathbb{R}_{\ge 0}\to\mathbb{R}_{\ge 0}$ is of class $\mathcal{K}$ if it is strictly increasing and satisfies $\alpha(0)=0$. It is of class $\mathcal{K}_\infty$ if, in addition, $\alpha(r)\to\infty$ as $r\to\infty$. A continuous function $\alpha:\mathbb{R}\to\mathbb{R}$ is an extended class $\mathcal{K}_\infty$ function, denoted by $\alpha\in\mathcal{K}_\infty^e$, if it is strictly increasing, satisfies $\alpha(0)=0$, and satisfies $\lim_{r\to\pm\infty}\alpha(r)=\pm\infty$.

\subsection{Safety and Control Barrier Functions}
We consider the nonlinear control-affine system
\begin{equation}
    \dot{x}=f(x)+g(x)u, \label{eq:nls}
\end{equation}
where $x\in\mathbb{R}^n$ is the state and $u\in\mathbb{R}^m$ is the control input. 
Throughout the paper, $f$ and $g$ are assumed to be smooth. For a locally Lipschitz state-feedback controller $u=k(x)$, the corresponding closed-loop system is
\begin{equation}
    \dot{x}=f_{\mathrm{cl}}(x):=f(x)+g(x)k(x). \label{eq:cl-nls}
\end{equation}
Let $h:\mathbb{R}^n\to\mathbb{R}$ be a continuously differentiable function, and define the set
\begin{equation}
    \mathcal{C}:=\{x\in\mathbb{R}^n:h(x)\ge 0\}.
\end{equation}
The set $\mathcal{C}$ is said to be \emph{forward invariant} for the closed-loop system~\eqref{eq:cl-nls} if, for every initial condition $x_0\in\mathcal{C}$, the corresponding trajectory satisfies $x(t)\in\mathcal{C}$ for all $t\ge 0$. The system is said to be \emph{safe} with respect to $\mathcal{C}$ if $\mathcal{C}$ is forward invariant~\citep{ames2016control,Ames2019Control}.

\begin{definition}[Zeroing CBF]\rm
A continuously differentiable function $h:\mathbb{R}^n\to\mathbb{R}$ is called a \textit{zeroing CBF} for system~\eqref{eq:nls} on an open set $\mathcal{S}\supset\mathcal{C}$, if there exists a function $\alpha_h\in\mathcal{K}_\infty^e$ such that
\begin{equation}
    \sup_{u\in\mathbb{R}^m}\bigl(L_fh(x)+L_gh(x)u\bigr)\ge -\alpha_h\bigl(h(x)\bigr),
    \quad \forall x\in\mathcal S.
\end{equation}
\end{definition}
The function $h$ is a zeroing CBF for system~\eqref{eq:nls} on $\mathcal S$ if and only if there exists a function $\alpha_h\in\mathcal{K}_\infty^e$ such that
\begin{equation}
    L_gh(x)=0 \implies L_fh(x)+\alpha_h\bigl(h(x)\bigr)\ge 0,
    \quad \forall x\in\mathcal S.
\end{equation}
In other words, a candidate barrier may be verified by examining the states at which the control input has no instantaneous effect on its derivative.

\begin{definition}[Reciprocal CBF]\rm~\label{def:rcbf}
A continuously differentiable function $B:\operatorname{Int}(\mathcal{C})\to\mathbb{R}$ is called a \textit{reciprocal CBF} for \eqref{eq:nls} on $\operatorname{Int}(\mathcal{C})$ if 
$B(x)>0$ for all $x\in\operatorname{Int}(\mathcal{C})$, $B(x)\to\infty$ as $x\to\partial \mathcal{C}$, and there exists a function $\alpha_B\in\mathcal{K}$  such that 
\begin{equation}\label{eq:60}
    \inf_{u\in\mathbb{R}^m}\Bigl(L_fB(x)+L_gB(x)u\Bigr)
    \le
    \alpha_B\!\left(\frac{1}{B(x)}\right),
    \quad \forall x\in\operatorname{Int}(\mathcal{C}).
\end{equation}
The function $B$ is a reciprocal CBF for system~\eqref{eq:nls} on $\operatorname{Int}(\mathcal{C})$ if and only if there exists a function $\alpha_B\in\mathcal{K}$ such that
\begin{equation}\label{eq:RCBF_def}
    L_gB(x)=0 \implies  L_fB(x)-\alpha_B\left(\frac{1}{B(x)}\right)\le 0, \quad\forall x\in\operatorname{Int}(\mathcal{C}).
\end{equation}
\end{definition}

\subsection{Nonholonomic Robot Model and Safety Constraint}

Consider a force-controlled nonholonomic robot described by
\begin{subequations}\label{eq:208}
    \begin{align}
        \dot{x} &= v \cos{\theta},\quad \dot{y} = v \sin{\theta},\quad \dot{\theta} = \omega, \\
        \dot v &= u_1, \quad \dot\omega = u_2,
    \end{align}
\end{subequations}
where $(x,y)\in\mathbb{R}^2$ denotes the robot position, $\theta\in\mathbb{R}$ is the heading angle, $v\in\mathbb{R}$ is the forward speed, $\omega\in\mathbb{R}$ is the angular speed, and $u=[u_1,u_2]^\top\in\mathbb{R}^2$ is the control input.
Denoting $\mathbf{x}:=[x,y,\theta,v,\omega]^\top$, \eqref{eq:208} can be written in the form of \eqref{eq:nls}, i.e.,  
\begin{equation}\label{eq:robots_force}
    \dot {\mathbf{x}} = f(\mathbf{x})+gu,
\end{equation}
where
\begin{equation}\label{eq:100}
    f(\mathbf{x})=
    \begin{bmatrix}
        v\cos\theta\\
        v\sin\theta\\
        \omega\\
        0\\
        0
    \end{bmatrix}
    \quad \text{and} \quad
    g=
    \begin{bmatrix}
        0&0\\
        0&0\\
        0&0\\
        1&0\\
        0&1
    \end{bmatrix}.
\end{equation}

We consider obstacle avoidance with respect to a circular obstacle centered at $c=[c_x,c_y]^\top\in\mathbb{R}^2$ with radius $R>0$. Let $p:=[x,y]^\top$. The associated physical safe set is given by
\begin{equation}\label{eq:safe_set_physical}
    \mathcal C_0:=\{\mathbf{x}\in\mathbb{R}^5:h_0(\mathbf{x}):=|p-c|^2-R^2\ge 0\}.
\end{equation}
Since $h_0$ depends only on the robot position $p$, the control input does not appear in its first time derivative. In particular, for system~\eqref{eq:robots_force}, one gets $L_g h_0(\mathbf{x}) = 0$.

To characterize the robot motion relative to the obstacle, define the radial distance and radial error, respectively, by
\begin{equation}\label{eq:r_er}
    r(\mathbf{x}):=|p-c| \quad\text{and}\quad e_r(\mathbf{x}):=r-R.
\end{equation}
Then, one has $h_0(\mathbf{x})=(r-R)(r+R)=e_r(r+R)$.
The radial velocity is given by
\begin{align}\label{eq:erdot_formulation}
    \dot e_r
    =
    \frac{(p-c)^\top}{|p-c|}
    \begin{bmatrix}
        \dot x\\
        \dot y
    \end{bmatrix} 
    =
    \frac{v}{r}[(x-c_x)\cos\theta+(y-c_y)\sin\theta].
\end{align}
On $\mathcal{C}_0$, one has $r\ge R>0$, so the radial coordinate is well defined.
The quantity $\dot{e}_r$ represents the radial component of the robot velocity relative to the obstacle and will be used to encode motion-dependent safety sensitivity in the barrier construction.

\subsection{Problem Statement}

For the mobile robot~\eqref{eq:robots_force}, the physical safety function $h_0$ cannot directly define a first-order CBF condition because the input does not appear in $\dot h_0$. The problem is therefore to recover first-order enforceability without replacing the original physical safe set by a smaller certified domain.

The objective of this paper is to construct a reciprocal CBF for obstacle avoidance whose domain is \textit{exactly} the interior of the physical safe set, namely $\operatorname{Int}(\mathcal C_0)$. This preserves the certified interior domain of the original physical safe set while allowing motion-dependent terms to restore first-order control authority. When a controller satisfies the induced reciprocal-CBF condition, the interior of the physical safe set is forward invariant. This condition can be implemented through optimization-based safety filters~\citep{ames2016control} or universal-formula-based controllers~\citep{wang2026universalformulafamiliessafe}.

%% -------------------------------------------------------------------
\section{Scaling-Based Reciprocal CBF Construction}\label{sec:main}

In this section, we construct a scaling-based reciprocal CBF for the force-controlled nonholonomic mobile robot and establish the associated forward-invariance result. 
Consider the system~\eqref{eq:robots_force}, together with the physical safe set~\eqref{eq:safe_set_physical} and its interior
$\operatorname{Int}(\mathcal{C}_0):=\{\mathbf{x}\in\mathbb{R}^5:h_0(\mathbf{x})>0\}$.
For~\eqref{eq:robots_force}, the physical safety function $h_0$ cannot be used directly to impose a first-order barrier condition because $L_g h_0(\mathbf{x})=0$.
A common way to recover first-order control authority is to modify the physical safety function by adding velocity-dependent terms. For instance, one may consider a barrier of the form $h(\mathbf{x})=h_0(\mathbf{x})-\eta(\mathbf{x})$, where $\eta(\mathbf{x})\ge 0$ depends on velocity-related states; related constructions appear in recursive, rectified, and backstepping-based CBF designs~\citep{Ong2024Rectified,Taylor2022Safe,Han2024Safety}. Such a modification can introduce direct input dependence into $\dot h$. However, since $h(\mathbf{x})\ge 0$ implies $h_0(\mathbf{x})\ge \eta(\mathbf{x})\ge 0$, the certified set associated with $h$ is, by construction, contained in the original physical safe set. Thus, this approach may preserve physical safety but can certify only a smaller domain than the one defined by $h_0$.

To preserve the certified domain of the physical safe set, one may first consider a positive state-dependent scaling of the original safety function,
\begin{equation}
    h(\mathbf{x}) := \sigma(\mathbf{x}) h_0(\mathbf{x}),
\end{equation}
where $\sigma(\mathbf{x})>0$ is bounded. Since $\sigma(\mathbf{x})$ is strictly positive, $h$ and $h_0$ have the same zero level set and therefore induce the same safe set. However, differentiating $h$ along the system trajectories gives
\begin{equation}
    \dot{h}
    =
    \underbrace{h_0(\mathbf{x}) L_f \sigma(\mathbf{x}) + \sigma(\mathbf{x}) L_f h_0(\mathbf{x})}_{L_f h(\mathbf{x})}
    +
    \underbrace{h_0(\mathbf{x}) L_g \sigma(\mathbf{x})}_{L_g h(\mathbf{x})} u .
\end{equation}
Thus, $L_g h(\mathbf{x})$ contains the factor $h_0(\mathbf{x})$. On the physical safety boundary, where $h_0(\mathbf{x})=0$, one has $L_g h(\mathbf{x})=0$. Hence, although the positive scaling preserves the certified domain, it does not restore first-order control authority on the boundary. This motivates the reciprocal construction introduced next.

To overcome this difficulty, we introduce the reciprocal CBF candidate
\begin{equation}\label{eq:B_robot_new}
    B(\mathbf{x})
    :=
    \frac{\lambda(\mathbf{x})}{h_0(\mathbf{x})},
\end{equation}
defined on $\operatorname{Int}(\mathcal{C}_0)$, where
\begin{equation}\label{eq:lambda_robot_new}
    \lambda(\mathbf{x})
    :=
    \varepsilon_v \arctan\!\big(k_v\,\dot{e}_r(\mathbf{x})\big)
    + \varepsilon_{\omega} \arctan\!\big(k_{\omega}\,\omega\big)
    + \pi,
\end{equation}
with $\varepsilon_v,\varepsilon_\omega\in(0,1)$ and $k_v,k_\omega>0$. This construction is motivated by two observations. First, $\lambda(\mathbf{x})$ is strictly positive and bounded on $\operatorname{Int}(\mathcal{C}_0)$, while $h_0(\mathbf{x})\to 0^+$ as $\mathbf{x}$ approaches $\partial\mathcal{C}_0$ from within $\operatorname{Int}(\mathcal{C}_0)$. Hence, $B(\mathbf{x})\to+\infty$ at the physical safety boundary, and the domain of $B$ is exactly the interior of the original physical safe set. Second, because $\lambda(\mathbf{x})$ depends on the motion-related variables $\dot e_r(\mathbf{x})$ and $\omega$, the derivative of $B$ can recover direct input dependence. The radial velocity $\dot e_r$ indicates whether the robot is instantaneously moving toward or away from the obstacle boundary, while the $\omega$-dependent term introduces direct dependence on $u_2$ through $\dot\omega=u_2$. This term also prevents loss of control authority in configurations where the radial-velocity contribution degenerates, such as when the robot heading is tangent to the obstacle-centered radial direction. The arctangent functions are used only as smooth bounded monotone maps, ensuring that $\lambda(\mathbf{x})$ remains positive and bounded while retaining nonzero input dependence.

The particular arctangent formula in~\eqref{eq:lambda_robot_new} is not essential. The key design principle is to place a positive, bounded, and motion-dependent scaling factor in the numerator of a reciprocal barrier. Indeed, the standard reciprocal barrier $1/h_0(\mathbf{x})$ does not recover input dependence when $L_g h_0(\mathbf{x})=0$, since $L_g\left({1}/{h_0(\mathbf{x})}\right)=0$.
Similarly, the direct multiplicative scaling $h(\mathbf{x})=\sigma(\mathbf{x})h_0(\mathbf{x})$ preserves the zero level set but gives $L_g h(\mathbf{x})=h_0(\mathbf{x})L_g\sigma(\mathbf{x})$,
which vanishes on the physical boundary $h_0(\mathbf{x})=0$. By contrast, the reciprocal quotient~\eqref{eq:B_robot_new} satisfies $L_g B(\mathbf{x})={L_g\lambda(\mathbf{x})}/{h_0(\mathbf{x})}$
whenever $L_g h_0(\mathbf{x})=0$. Thus, the numerator scaling restores first-order control authority through $L_g\lambda(\mathbf{x})$, while the positivity and boundedness properties of $\lambda(\mathbf{x})$ ensure that the barrier remains defined exactly on $\operatorname{Int}(\mathcal C_0)$ and blows up at the physical boundary.
More generally, the scaling function may be designed by identifying motion-related variables that characterize the evolution of the physical safety constraint and introduce control authority through $L_g\lambda$. For a relative-degree-$r$ safety constraint, $L_f^{r-1}h_0$ provides a natural starting point. The resulting $\lambda(x)$ should remain continuously differentiable, strictly positive, and bounded, with $L_g\lambda(x)\neq0$ on the intended certified domain. If one motion-dependent term degenerates in certain configurations, complementary variables associated with other input channels may be included.

The following result establishes that the proposed candidate defines a reciprocal CBF for the force-controlled nonholonomic mobile robot.

\begin{theorem}[Scaling-based reciprocal CBF]
\label{thm:scbf}
Consider the force-controlled nonholonomic mobile robot system~\eqref{eq:robots_force} and the physical safe set~\eqref{eq:safe_set_physical}, with interior $\operatorname{Int}(\mathcal{C}_0)$.
Let $B$ be defined by~\eqref{eq:B_robot_new}, where the scaling factor $\lambda$ is given by~\eqref{eq:lambda_robot_new}. Then $B$ is well defined and strictly positive on $\operatorname{Int}(\mathcal{C}_0)$, and satisfies
\begin{equation}\label{eq:B_blowup_boundary}
    \lim_{\mathbf{x}\to \partial \mathcal{C}_0}
        B(\mathbf{x})=+\infty.
\end{equation}
Furthermore, $B$ is a reciprocal CBF on $\operatorname{Int}(\mathcal{C}_0)$.
\end{theorem}

\begin{proof}
Since $\varepsilon_v,\varepsilon_\omega\in(0,1)$ and $\arctan(\cdot)\in(-\pi/2,\pi/2)$, the scaling factor $\lambda$ satisfies
\begin{equation}
    \lambda(\mathbf{x})
    >
    \pi-\frac{\pi}{2}(\varepsilon_v+\varepsilon_\omega)
    >0.
\end{equation}
Moreover, the arctangent terms are bounded. Hence, there exist constants $0<\underline{\lambda}\le \overline{\lambda}<\infty$ such that $\underline{\lambda}\le \lambda(\mathbf{x})\le \overline{\lambda}$ for all $\mathbf{x}\in\operatorname{Int}(\mathcal C_0)$. Since $h_0(\mathbf{x})>0$ on $\operatorname{Int}(\mathcal C_0)$, the function $B(\mathbf{x})=\lambda(\mathbf{x})/h_0(\mathbf{x})$ is well defined, continuously differentiable, and strictly positive on $\operatorname{Int}(\mathcal C_0)$. Furthermore, as $\mathbf{x}$ approaches $\partial\mathcal C_0$ from within $\operatorname{Int}(\mathcal C_0)$, one has $h_0(\mathbf{x})\to 0^+$, while $\lambda(\mathbf{x})\ge\underline{\lambda}>0$. Therefore, $B(\mathbf{x})\to+\infty$.

It remains to verify the reciprocal-CBF condition. Since $L_g h_0(\mathbf{x})=0$, differentiating $B=\lambda/h_0$ gives
\begin{equation}\label{eq:LgB_robot}
    L_gB(\mathbf{x})
    =
    \frac{L_g\lambda(\mathbf{x})}{h_0(\mathbf{x})}.
\end{equation}
For the input matrix $g$ in~\eqref{eq:100}, one has $L_g\lambda(\mathbf{x})=[\partial\lambda/\partial v,\ \partial\lambda/\partial\omega]$. From~\eqref{eq:lambda_robot_new},
\begin{equation}
    \frac{\partial\lambda}{\partial\omega}
    =
    \frac{\varepsilon_\omega k_\omega}{1+k_\omega^2\omega^2}
    >0.
\end{equation}
Thus, $L_g\lambda(\mathbf{x})\neq0$ for all $\mathbf{x}\in\operatorname{Int}(\mathcal C_0)$. Since $h_0(\mathbf{x})>0$ on $\operatorname{Int}(\mathcal C_0)$, it follows that $L_gB(\mathbf{x})\neq0$ on $\operatorname{Int}(\mathcal C_0)$. Therefore, the implication condition in Definition~\ref{def:rcbf} is satisfied vacuously for any $\alpha_B\in\mathcal K$. Hence, $B$ is a reciprocal CBF on $\operatorname{Int}(\mathcal C_0)$.
\end{proof}

\begin{corollary}
\label{cor:forward-invariance}
Suppose that $\mathbf{x}(0)\in\operatorname{Int}(\mathcal C_0)$ and that the closed-loop system admits a locally absolutely continuous solution on its maximal interval of existence. If the feedback controller $u=k(\mathbf{x})$ satisfies
\begin{equation}
    L_fB(\mathbf{x})+L_gB(\mathbf{x})k(\mathbf{x})
    \le
    \alpha_B\left(\frac{1}{B(\mathbf{x})}\right),
    \quad
    \mathbf{x}\in\operatorname{Int}(\mathcal C_0),
\end{equation}
then the boundary $\partial\mathcal C_0$ cannot be reached in finite time. Consequently, the solution remains in $\operatorname{Int}(\mathcal C_0)$ throughout its interval of existence. If the closed-loop solution is forward complete, then $\operatorname{Int}(\mathcal C_0)$ is forward invariant for all $t\ge 0$.
\end{corollary}

\begin{proof}
Let $\mathbf{x}(t)$ be a locally absolutely continuous closed-loop solution and define $B(t):=B(\mathbf{x}(t))$. On any interval where $\mathbf{x}(t)\in\operatorname{Int}(\mathcal C_0)$, the function $B(t)$ is locally absolutely continuous. By the assumed reciprocal-CBF inequality,
\begin{equation}
    \dot B(t)
    \le
    \alpha_B\left(\frac{1}{B(t)}\right)
\end{equation}
for almost all such $t$.

Suppose, for contradiction, that the trajectory reaches $\partial\mathcal C_0$ at a finite first time $T>0$. Then $\mathbf{x}(t)\in\operatorname{Int}(\mathcal C_0)$ for all $t\in[0,T)$, and the boundary blow-up property of $B$ implies
\begin{equation}\label{eq:424}
    \lim_{t\to T^-}B(t)=+\infty.
\end{equation}
Let $M:=\max\{B(0),1\}$ and $c_M:=\alpha_B(1/M)$. Whenever $B(t)\ge M$, monotonicity of $\alpha_B$ gives
\begin{equation}
    \dot B(t)
    \le
    \alpha_B\left(\frac{1}{B(t)}\right)
    \le
    \alpha_B\left(\frac{1}{M}\right)
    =
    c_M
\end{equation}
for almost all such $t$. Therefore, once $B(t)$ is above $M$, it can grow at most linearly. More precisely, for any $t<T$ with $B(t)>M$, let $\tau\in[0,t]$ be the last time before $t$ such that $B(\tau)=M$, or set $\tau=0$ if $B(s)\ge M$ for all $s\in[0,t]$. Integrating on $[\tau,t]$ yields
\begin{equation}
    B(t)\le M+c_M(t-\tau)\le M+c_MT.
\end{equation}
If $B(t)\le M$, it is already bounded. Hence $B(t)$ is bounded on $[0,T)$, contradicting \eqref{eq:424}. Thus the boundary cannot be reached in finite time, and the solution remains in $\operatorname{Int}(\mathcal C_0)$ on its interval of existence. If the solution is forward complete, the conclusion holds for all $t\ge0$.
\end{proof}

The construction above preserves the certified interior domain of the physical safe set because the reciprocal barrier is defined on $\operatorname{Int}(\mathcal C_0)$. This domain-preservation property does not imply improved robustness, reduced control effort, or feasibility under bounded actuation. In particular, states close to the physical safety boundary or with unfavorable motion may require large corrective inputs. Thus, preservation of the full certified domain may shift conservatism from restriction of the certified state domain to increased control demand.

\begin{remark}\rm
The construction in Theorem~\ref{thm:scbf} is not unique. The specific arctangent formula in~\eqref{eq:lambda_robot_new} is one convenient choice of a positive, bounded, motion-dependent scaling factor. More generally, any $C^1$ function $\lambda$ satisfying $0<\underline{\lambda}\le \lambda(\mathbf{x})\le\overline{\lambda}<\infty$ and $L_g\lambda(\mathbf{x})\neq0$ on $\operatorname{Int}(\mathcal C_0)$ can be used in the same reciprocal construction.
The quotient form in~\eqref{eq:B_robot_new} can also be generalized. Consider
\begin{equation}
    B_{\psi}(\mathbf{x})
    :=
    \psi\!\left(\frac{\lambda(\mathbf{x})}{h_0(\mathbf{x})}\right),
\end{equation}
where $\psi:(0,\infty)\to(0,\infty)$ is $C^1$, unbounded, and satisfies $\psi'(s)>0$ for all $s>0$. Then $B_{\psi}$ is well defined on $\operatorname{Int}(\mathcal C_0)$ and satisfies $B_{\psi}(\mathbf{x})\to+\infty$ as $\mathbf{x}$ approaches $\partial\mathcal C_0$ from within $\operatorname{Int}(\mathcal C_0)$. Since $L_g h_0(\mathbf{x})=0$, one has
\begin{equation}
    L_g B_{\psi}(\mathbf{x})
    =
    \psi'\!\left(\frac{\lambda(\mathbf{x})}{h_0(\mathbf{x})}\right)
    \frac{L_g\lambda(\mathbf{x})}{h_0(\mathbf{x})}.
\end{equation}
Thus, $L_gB_{\psi}(\mathbf{x})\neq0$ whenever $L_g\lambda(\mathbf{x})\neq0$.

As a logarithmic example, one may take
\begin{equation}
    B_{\ln}(\mathbf{x})
    :=
    \ln\!\left(1+\frac{\lambda(\mathbf{x})}{h_0(\mathbf{x})}\right),
\end{equation}
for which
\begin{equation}
    L_g B_{\ln}(\mathbf{x})
    =
    \frac{L_g\lambda(\mathbf{x})}{h_0(\mathbf{x})+\lambda(\mathbf{x})}.
\end{equation}
This logarithmic transformation preserves the same domain and boundary blow-up properties while softening the growth of the barrier value and the input coefficient compared with the quotient form~\eqref{eq:B_robot_new}.
Other logarithmic or additive variants may also be constructed, provided that positivity, boundary blow-up, and nonvanishing input dependence are maintained. A systematic classification of such alternatives is left for future work.
\end{remark}

\begin{remark}\rm
The reciprocal CBF $B$ is defined on $\operatorname{Int}(\mathcal C_0)$ rather than on the closed set $\mathcal C_0$, since $B(\mathbf{x})$ becomes singular as $h_0(\mathbf{x})\to0^+$. Therefore, the safety guarantee should be interpreted as forward invariance of the interior $\operatorname{Int}(\mathcal C_0)$, or equivalently, non-reachability of the boundary $\partial\mathcal C_0$ for trajectories initialized in $\operatorname{Int}(\mathcal C_0)$. Thus, the proposed construction preserves the certified interior domain of the physical safe set, but does not assert that $B$ is well defined on the boundary.
\end{remark}

\begin{remark}\rm 
The result above is established under the unconstrained-input setting $u\in\mathbb R^2$. If the input is restricted to a compact set $\mathcal U\subset\mathbb R^2$, pointwise feasibility requires
\begin{equation}
    \inf_{u\in\mathcal U}
    \left(
        L_fB(\mathbf{x})+L_gB(\mathbf{x})u
    \right)
    \le
    \alpha_B\left(\frac{1}{B(\mathbf{x})}\right),
    \quad
    \mathbf{x}\in\operatorname{Int}(\mathcal C_0),
\end{equation}
which is not guaranteed by $L_gB(\mathbf{x})\neq0$ alone. Moreover, because $L_fB$ and $L_gB$ contain singular terms in $h_0(\mathbf{x})$, the required control action may become large as $h_0(\mathbf{x})\to0^+$. Feasibility under actuator saturation therefore requires additional analysis beyond the scope of this paper.
\end{remark}

%% -------------------------------------------------------------------
\section{Scalar Strict-Feedback Illustration}
\label{sec:strict-feedback-system}

The purpose of this section is not to develop a general extension to arbitrary higher-relative-degree or nonholonomic systems, but to provide a structural interpretation of the mechanism underlying the force-controlled nonholonomic robot considered in Section~\ref{sec:main}. The robot possesses an inherent kinematic--dynamic cascade: the position-level safety constraint evolves through motion-related states, while the control inputs enter through higher-order velocity dynamics. A scalar strict-feedback system provides a canonical setting in which this cascade structure and the role of the motion-dependent numerator scaling can be examined independently of the particular robot geometry.
For a first-state safety constraint of relative degree $n$, the quantity $L_f^{n-1}h_0$ is the highest drift derivative before the input appears explicitly in the next derivative. Thus, choosing the scaling factor as a positive bounded function of $L_f^{n-1}h_0$ provides a direct analogue of the motion-dependent numerator scaling used in Section~\ref{sec:main}. 
This viewpoint further shows that the proposed scaling mechanism is associated with the underlying higher-relative-degree cascade structure, rather than being tied to the specific obstacle geometry of the robot example.

We consider the following scalar strict-feedback system:
\begin{equation}
\begin{aligned}
\dot{x}_1 &= f_1(x_1)+g_1(x_1)x_2,\\
\dot{x}_2 &= f_2(x_1,x_2)+g_2(x_1,x_2)x_3,\\
&\ \vdots\\
\dot{x}_{n-1} &= f_{n-1}(x_1,\dots,x_{n-1})
+g_{n-1}(x_1,\dots,x_{n-1})x_n,\\
\dot{x}_n &= f_n(x_1,\dots,x_n)+g_n(x_1,\dots,x_n)u,
\end{aligned}
\label{eq:strict_feedback}
\end{equation}
where $x=(x_1,\dots,x_n)^\top\in\mathbb{R}^n$ and $u\in\mathbb{R}$. The functions $f_i$ and $g_i$ are assumed to be smooth on their domains, and the input-channel functions are assumed to be nonvanishing, namely,
\begin{equation}
    g_i(x_1,\dots,x_i)\neq 0,\qquad i=1,\dots,n.
\label{eq:nonvanishing_g}
\end{equation}

\begin{theorem}
\label{Thm:RCBF_SF}
Consider the scalar strict-feedback system~\eqref{eq:strict_feedback} under the nonvanishing input-channel condition~\eqref{eq:nonvanishing_g}. Let $h_0:\mathbb R^n\to\mathbb R$ be a physical safety function, and let $\mathcal C_0:=\{x\in\mathbb R^n:h_0(x)\ge0\}$ with $\operatorname{Int}(\mathcal C_0)=\{x:h_0(x)>0\}$. Assume that:
\begin{enumerate}
    \item[(i)] $h_0(x)=\bar h(x_1)$ for some $\bar h\in C^n(\mathbb R)$.

    \item[(ii)] $h_0$ has relative degree $n$ with respect to the input on $\operatorname{Int}(\mathcal C_0)$; that is, $L_gL_f^i h_0(x)=0$ for $i=0,\dots,n-2$, and
    \begin{equation}
        L_gL_f^{n-1}h_0(x)\neq0,
        \quad
        x\in\operatorname{Int}(\mathcal C_0).
    \label{eq:sf_relative_degree_nonzero}
    \end{equation}

    \item[(iii)] With $\Lambda(x):=L_f^{n-1}h_0(x)$, the function $\phi:\mathbb R\to(0,\infty)$ is $C^1$, and there exist constants $0<\underline\phi\le\overline\phi<\infty$ such that $\underline\phi\le\phi(\Lambda(x))\le\overline\phi$ for all $x\in\operatorname{Int}(\mathcal C_0)$.

    \item[(iv)] The scaling function satisfies $\phi'(\Lambda(x))\neq0$ for all $x\in\operatorname{Int}(\mathcal C_0)$.
\end{enumerate}
Define
\begin{equation}
    B(x)
    :=
    \frac{\phi(\Lambda(x))}{h_0(x)},
    \quad
    x\in\operatorname{Int}(\mathcal C_0).
\label{eq:sf_RCBF}
\end{equation}
Then, $B$ is well-defined, continuously differentiable, and strictly positive on $\operatorname{Int}(\mathcal C_0)$, and satisfies
\begin{equation}
    \lim_{x\to\partial\mathcal C_0}
    B(x)=+\infty.
\label{eq:sf_B_blowup}
\end{equation}
Moreover,
\begin{equation}
    L_gB(x)
    =
    \frac{
    \phi'(\Lambda(x))L_gL_f^{n-1}h_0(x)
    }{
    h_0(x)
    }
    \neq 0,
    \quad
    x\in\operatorname{Int}(\mathcal C_0).
\label{eq:sf_LgB}
\end{equation}
Consequently, under unconstrained inputs, $B$ is a reciprocal CBF on $\operatorname{Int}(\mathcal C_0)$.
\end{theorem}

\begin{proof}
Since $h_0\in C^n$, the vector fields of~\eqref{eq:strict_feedback} are smooth, and $\phi\in C^1$, the function $\Lambda=L_f^{n-1}h_0$ is $C^1$, and hence $B$ is $C^1$ on $\operatorname{Int}(\mathcal C_0)$. Moreover, $h_0(x)>0$ and $\phi(\Lambda(x))\ge\underline\phi>0$ on $\operatorname{Int}(\mathcal C_0)$. Thus, $B$ is well defined and strictly positive on $\operatorname{Int}(\mathcal C_0)$.

As $x$ approaches $\partial\mathcal C_0$ from within $\operatorname{Int}(\mathcal C_0)$, one has $h_0(x)\to0^+$. Since $\phi(\Lambda(x))\ge\underline\phi$, it follows that
\begin{equation}
    B(x)
    =
    \frac{\phi(\Lambda(x))}{h_0(x)}
    \ge
    \frac{\underline\phi}{h_0(x)}
    \to+\infty.
\end{equation}
This proves the boundary blow-up property.

It remains to verify direct input dependence. Since $h_0$ has relative degree $n$, one has $L_gh_0(x)=0$ on $\operatorname{Int}(\mathcal C_0)$. Differentiating $B=\phi(\Lambda)/h_0$ along the input vector field gives
\begin{equation}
    L_gB(x)
    =
    \frac{\phi'(\Lambda(x))L_g\Lambda(x)}{h_0(x)}
    =
    \frac{\phi'(\Lambda(x))L_gL_f^{n-1}h_0(x)}{h_0(x)}.
\end{equation}
By assumptions~(ii) and~(iv), and since $h_0(x)>0$ on $\operatorname{Int}(\mathcal C_0)$, we obtain $L_gB(x)\neq0$ for all $x\in\operatorname{Int}(\mathcal C_0)$.

Therefore, under unconstrained inputs, the reciprocal-CBF implication condition in Definition~\ref{def:rcbf} is satisfied vacuously. Hence, $B$ is a reciprocal CBF on $\operatorname{Int}(\mathcal C_0)$.
\end{proof}

\begin{remark}\rm
\label{rem:sf_admissible_scaling}
The assumptions on $\phi$ in Theorem~\ref{Thm:RCBF_SF} are readily satisfied. For example, one may choose
\begin{equation}
    \phi(s)=\phi_0+\varepsilon\arctan(ks),
\label{eq:sf_arctan_phi}
\end{equation}
where $\varepsilon>0$, $k>0$, and $\phi_0>\varepsilon\pi/2$. Since $\arctan(ks)\in(-\pi/2,\pi/2)$, this choice satisfies $0<\phi_0-\varepsilon\pi/2\le \phi(s)\le \phi_0+\varepsilon\pi/2<\infty$ for all $s\in\mathbb R$. Moreover, $\phi'(s)=\varepsilon k/(1+k^2s^2)>0$ for all $s\in\mathbb R$. Hence, this arctangent-type scaling satisfies the positivity, boundedness, and nonvanishing-derivative requirements in Theorem~\ref{Thm:RCBF_SF}, and is consistent with the motion-dependent scaling used in~\eqref{eq:lambda_robot_new}.
\end{remark}

%% -------------------------------------------------------------------
\section{Numerical Examples}\label{sec:simulations}

This section presents two numerical examples. The first example uses a double-integrator system to illustrate how different barrier constructions affect the certified safe domain. The second example applies the proposed scaling-based reciprocal CBF to obstacle avoidance for a force-controlled nonholonomic mobile robot.

\subsection{Certified-Domain Comparison: Double-Integrator System}

We first consider a simple double-integrator system as a geometric illustration. This reduced example is used to visualize how different barrier constructions affect the certified safety region, rather than to provide a comprehensive closed-loop performance comparison. In particular, it shows that, in this example, additive or recursive constructions can produce certified domains whose boundaries differ from the original physical safety specification, whereas the proposed scaling-based reciprocal construction preserves the certified interior domain of the original physical safe set.

We consider the double-integrator system with state $x=(x_1,x_2)\in\mathbb{R}^2$,
\begin{equation}
    \dot{x}_1 = x_2,\qquad \dot{x}_2 = u,
\end{equation}
where $x_1$ denotes the position and $x_2$ denotes the velocity. The safety requirement is imposed on the position coordinate through the physical safety function
\begin{equation}
    h_0(x) = 1 - x_1^2.
\end{equation}
The corresponding physical safe set is
\begin{equation}
    \mathcal{C}_0
    =
    \{x\in\mathbb{R}^2:h_0(x)\ge 0\}
    =
    \{(x_1,x_2)\in\mathbb{R}^2:|x_1|\le 1\}.
\end{equation}
Thus, the position is constrained to remain in the interval $[-1,1]$, while the velocity is unrestricted.

Figure~\ref{fig:compare} compares the certified safe regions induced by several CBF constructions in the $(x_1,x_2)$ plane. The green region represents the physical safe set $\mathcal C_0$, the red regions denote physically unsafe states, and the black vertical lines indicate the physical boundary $\partial\mathcal C_0$, namely $x_1=\pm1$. For the HOCBF, ReCBF, and backstepping-based constructions, the plotted certified domain is the subset of $\mathcal C_0$ satisfying the corresponding auxiliary barrier condition. The constructions used for comparison are summarized below.

\textit{1) HOCBF.}
The first HOCBF auxiliary function is chosen as
\begin{equation}
    H_{\mathrm{HOCBF}}(x) := \dot h_0(x)+\alpha(h_0(x)),
\end{equation}
where $\alpha(s)=\gamma_0s$, $\gamma_0>0$, and $\dot h_0(x)=-2x_1x_2$. Thus, $H_{\mathrm{HOCBF}}(x)=-2x_1x_2+\gamma_0(1-x_1^2)$.
As shown in Fig.~\ref{fig:compare}, the condition $H_{\mathrm{HOCBF}}(x)\ge0$ produces a tilted certified domain in the $(x_1,x_2)$ plane. In this example, the resulting certificate covers only a subset of the physical safe set $\mathcal C_0$.

\textit{2) ReCBF.}
The rectified construction uses an activation term of the form
\begin{equation}
    H_{\mathrm{ReCBF}}(x)
    =
    h_0(x)-\mu\,\mathrm{ReLU}(-s(x)),
\end{equation}
where $\mu>0$, $\mathrm{ReLU}(r):=\max\{0,r\}$, and $s(x)=\dot h_0(x)+\gamma_0h_0(x)$. The rectification term is inactive when $s(x)\ge0$ and becomes active when $s(x)<0$. As illustrated in Fig.~\ref{fig:compare}, the resulting certified domain remains inside the physical safe set and, for the parameters used here, is less restrictive than the HOCBF-based certified domain.

\textit{3) Backstepping-based CBF.}
The backstepping-based construction is taken as
\begin{equation}
    H_{\mathrm{Backstepping}}(x)
    =
    h_0(x)-\frac{1}{2\mu}\big(x_2-\kappa(x_1)\big)^2,
\end{equation}
where $\mu>0$ and the virtual control is chosen as $\kappa(x_1)=-kx_1$ with $k>0$. This construction certifies a compact subset of $\mathcal C_0$ in the $(x_1,x_2)$ plane. Therefore, while the certified states remain physically safe, the certified domain does not coincide with the full physical safe set.

\textit{4) Scaling-based reciprocal CBF.}
To illustrate the proposed scaling idea in this simple setting, we consider the reciprocal barrier
\begin{equation}
    B(x)=\frac{\lambda_0+\varepsilon \arctan(k_v x_2)}{h_0(x)},
\end{equation}
where $\lambda_0>\varepsilon\pi/2$, $\varepsilon>0$, and $k_v>0$, so that the numerator remains strictly positive. As shown in Fig.~\ref{fig:compare}, this construction highlights the key geometric feature of the proposed approach: the reciprocal barrier is modulated in a state-dependent manner without altering the original physical safety boundary. Consequently, the certified domain of the barrier coincides with the interior of the physical safe set.

For the certified-domain comparison in Fig.~\ref{fig:compare}, the parameters are chosen as
$\gamma_0=1$ for the HOCBF;
$\gamma_0=1$ and $\mu=1$ for the ReCBF;
$\mu=1$ and $k=1$ for the backstepping-based CBF; and
$\lambda_0=2$, $\varepsilon=0.5$, and $k_v=0.3$ for the proposed scaling-based reciprocal CBF.

For the closed-loop illustration in Fig.~\ref{fig:compare_2}, all methods use the common nominal controller
$u_{\mathrm{nom}}=-k_p x_1-k_d x_2$, with $k_p=2$ and $k_d=2.5$.
At each integration step, the control input is obtained from the scalar safety-filter problem
$u^\star = \arg\min_{u\in\mathbb R} \frac{1}{2}(u-u_{\rm nom})^2$ subject to the corresponding barrier constraint.
For the closed-loop simulations, the parameters are selected as follows:
$\gamma_0=2$ and $\gamma_1=3$ for the HOCBF;
$\gamma_0=1$, $\mu=0.7$, and $\alpha_H(s)=2s$ for the ReCBF;
$k=1$, $\mu=0.18$, and $\alpha_H(s)=2s$ for the backstepping-based CBF; and
$\lambda_0=2$, $\varepsilon=0.5$, $k_v=0.3$, and $k_B=2$ for the proposed reciprocal CBF, with $\alpha_B(1/B)=k_B/B$.
The trajectories are simulated over $4$~s using explicit Euler integration with a step size of $0.001$~s, without input saturation.
Figure~\ref{fig:compare_2} complements the certified-domain comparison in Fig.~\ref{fig:compare} by illustrating the corresponding closed-loop behavior under a common nominal controller and method-specific barrier parameters.
The common initial condition is
$(x_1(0),x_2(0))=(0.8,2.5)$.
For the parameters used in the closed-loop simulations, this initial condition lies inside the physical safe set $C_0$, but outside the certified domains associated with the HOCBF, ReCBF, and backstepping-based constructions.
Therefore, the corresponding trajectories should not be interpreted as a comprehensive closed-loop performance comparison or as evidence of general superiority or inferiority among the methods.
Rather, they illustrate the consequence of certified-domain shrinkage: when an initial condition is not covered by a method's certified domain, the associated safety certificate does not apply from that state.

By contrast, for the proposed scaling-based reciprocal CBF, the certified domain coincides with $\operatorname{Int}(C_0)$.
Hence, the same initial condition is certified by the proposed construction, and the corresponding trajectory remains within the physical safe set.
Thus, Fig.~\ref{fig:compare_2} provides a closed-loop illustration of the certified-domain distinction shown in Fig.~\ref{fig:compare}, rather than claiming a universal performance advantage under arbitrary initial conditions or controller tunings.

\begin{figure}
    \centering
    \includegraphics[width=0.8\linewidth]{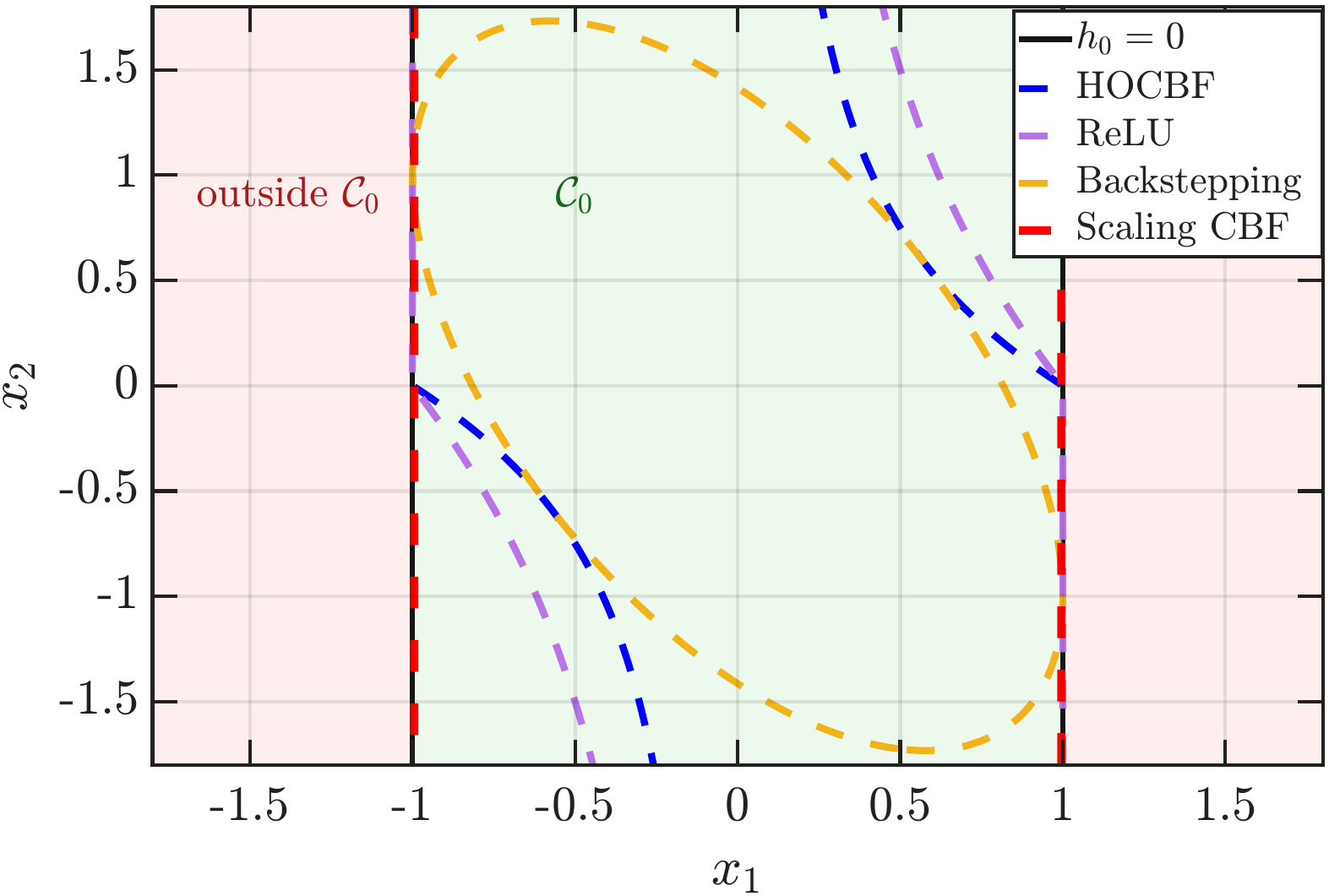}
    \caption{Certified-domain comparison for different CBF constructions in the $(x_1,x_2)$ plane. The HOCBF, ReCBF, and backstepping-based constructions certify subsets of $\mathcal C_0$, whereas the proposed scaling-based reciprocal CBF preserves the certified interior domain of the physical safe set.}
\label{fig:compare}
\end{figure}

\begin{figure}[t]
    \centering
    \includegraphics[width=0.8\linewidth]{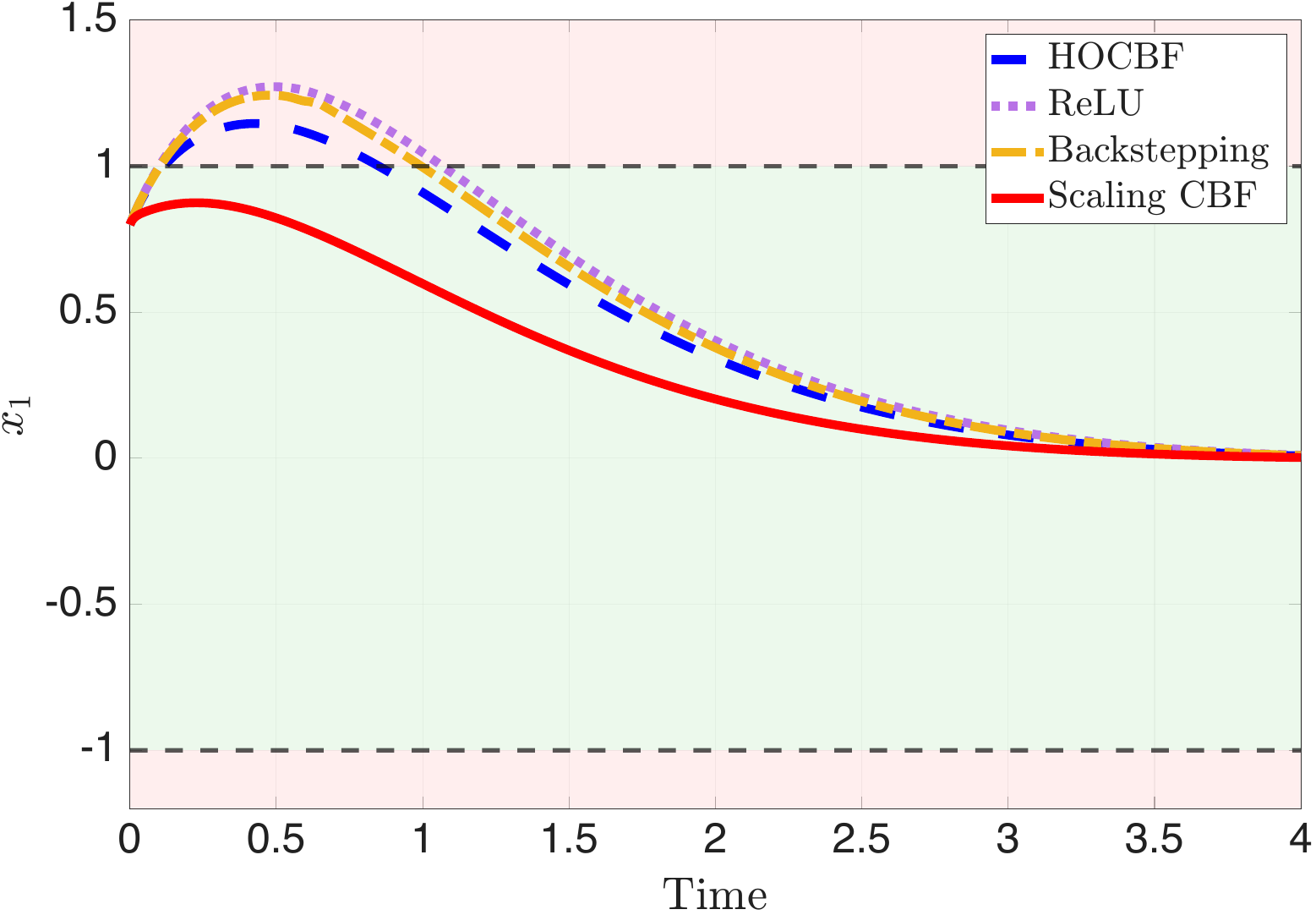}
    \caption{Illustrative position trajectories of the double-integrator system from $(x_1(0),x_2(0))=(0.8,2.5)$. The plot illustrates certified-domain effects rather than a general closed-loop performance comparison.}
\label{fig:compare_2}
\end{figure}

%% -------------------------------------------------------------------
\subsection{Obstacle Avoidance for a Force-Controlled Nonholonomic Robot}

We next illustrate the proposed scaling-based reciprocal CBF on the force-controlled nonholonomic mobile robot system~\eqref{eq:robots_force}. The purpose of this example is to demonstrate implementation on the full robot dynamics and integration with an optimization-based safety controller.

The control objective is to drive the robot toward the origin while ensuring collision avoidance with a circular obstacle centered at $c=[2,2]^\top$ with radius $R=1$. The corresponding physical safety function is $h_0(\mathbf{x}) = (x-2)^2 + (y-2)^2 - 1$.
The reciprocal CBF $B$ is constructed as in~\eqref{eq:B_robot_new}, with the scaling factor $\lambda(\mathbf{x})$ given by~\eqref{eq:lambda_robot_new}. The design parameters are selected as $\varepsilon_v = 0.8, \varepsilon_\omega = 0.15, k_v = 3.2, k_\omega = 0.3$.
The corresponding CBF constraint is imposed in the form
\begin{equation}
    L_f B(\mathbf{x})-\alpha_B\!\left(\frac{1}{B(\mathbf{x})}\right) + L_g B(\mathbf{x})u \le 0,
\end{equation}
where the class-$\mathcal{K}$ function is chosen as $\alpha_B ({1}/{B(\mathbf{x})})={k_B}/{B(\mathbf{x})}$ with $k_B = 2$. 
The nominal stabilizing controller and CLF--CBF--QP implementation follow~\citep{Han2024Safety}, with the proposed reciprocal-CBF constraint replacing the original safety constraint and all other controller and QP parameters unchanged.
The simulations are performed in MATLAB/Simulink over a time horizon of $20$~s using the fixed-step fourth-order Runge--Kutta (\texttt{ode4}) solver with a step size of $0.001$~s.

This nonholonomic robot example is not intended as a full benchmark comparison with HOCBF, ReCBF, or backstepping-based safety filters. Rather, it illustrates that the proposed scaling-based reciprocal CBF can be implemented on the full force-controlled nonholonomic dynamics within an optimization-based safety controller. The robot is simulated from the five initial conditions in Table~\ref{tab:nh_initial_conditions}, chosen to represent different approach directions and configurations relative to the obstacle and the goal. These cases do not imply that the initial conditions lie outside the certified domains of all benchmark methods. Since no actuator saturation is imposed, the reported inputs should be interpreted as requested safety-filter inputs rather than as a bounded-input feasibility certificate. The resulting trajectories are shown in Fig.~\ref{fig:wmr_traj_sin}. For all tested initial conditions, the robot successfully avoids the circular obstacle while moving toward the origin.

\begin{table}[t]
\caption{Initial conditions for the nonholonomic robot simulations.}
\label{tab:nh_initial_conditions}
\centering
\renewcommand{\arraystretch}{1.1}
\begin{tabular}{c c c c c c}
\hline
Case & $x(0)$ & $y(0)$ & $\theta(0)$ & $v(0)$ & $\omega(0)$ \\
\hline
1 & $4.30$ & $2.60$ & $0.95$ & $0.73$ & $0.53$ \\
2 & $0.30$ & $2.40$ & $1.50$ & $0.12$ & $0.12$ \\
3 & $4.00$ & $3.30$ & $1.70$ & $0.10$ & $0.22$ \\
4 & $4.80$ & $3.40$ & $1.60$ & $0.82$ & $0.26$ \\
5 & $4.50$ & $1.80$ & $0.40$ & $0.52$ & $0.37$ \\
\hline
\end{tabular}
\end{table}

\begin{figure}[t]
    \centering
    \includegraphics[width=0.8\linewidth]{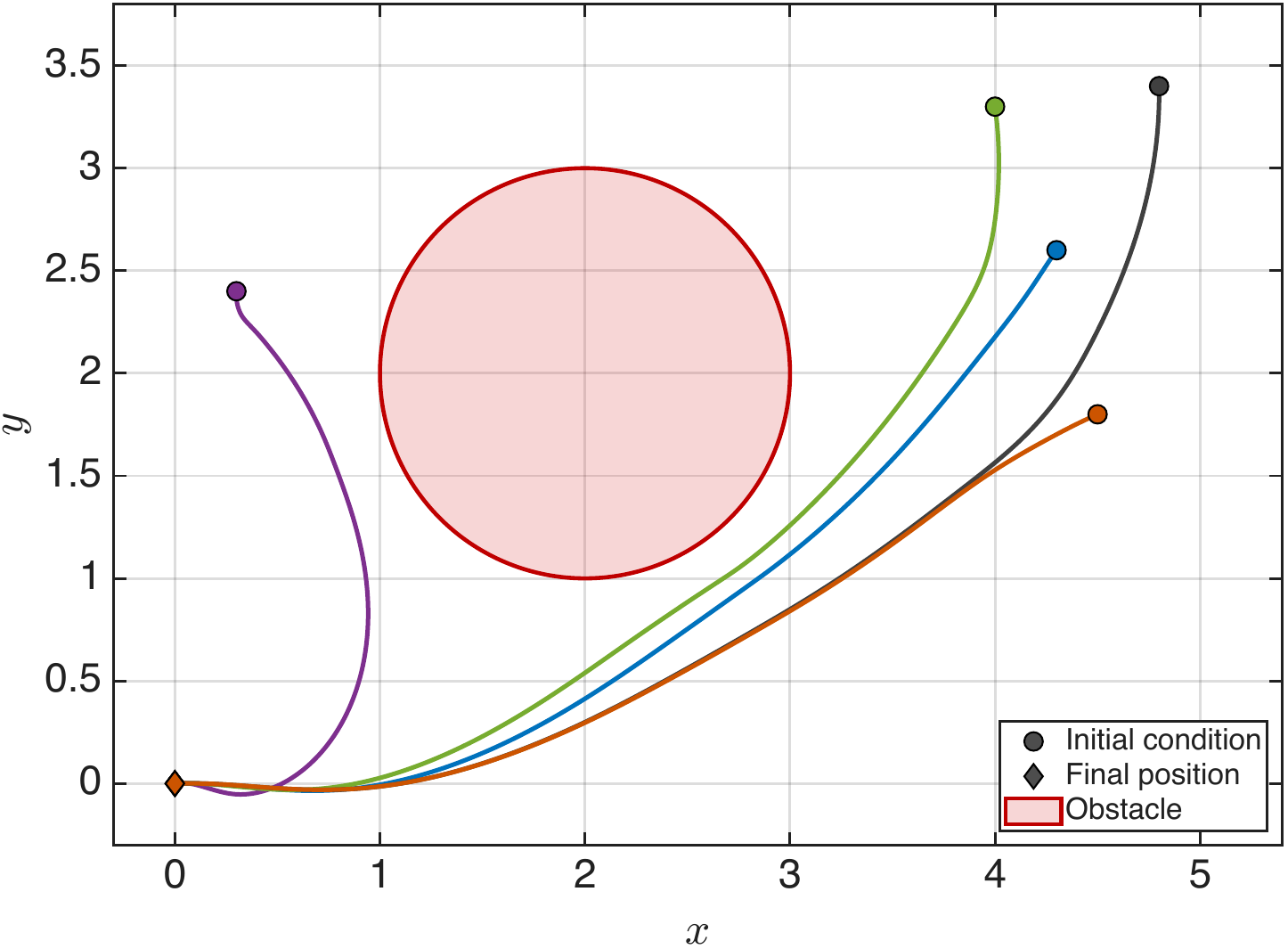}
    \caption{Closed-loop trajectories of the force-controlled nonholonomic robot under the optimization-based safety controller with the proposed scaling-based reciprocal CBF.}
    \label{fig:wmr_traj_sin}
\end{figure}

Table~\ref{tab:nh_summary} summarizes the safety diagnostics and requested inputs for the five nonholonomic robot simulations. In all cases, the signed clearance $d(t)$ and the physical safety function $h_0(t)$ remain positive, so the robot does not reach the obstacle boundary. The smallest clearance is $0.2064$, and the smallest value of $h_0(t)$ is $0.4553$, both occurring in Case~3.
The input summary indicates that the safety filter may request large transient inputs. Across the five simulations, the largest requested inputs are $\max |u_1(t)|=149.24$ and $\max |u_2(t)|=91.24$, both occurring in Case~4. Since no actuator saturation is imposed, these values quantify the requested inputs of the safety filter but do not imply feasibility under bounded actuation. 
These results illustrate the trade-off associated with certificate-domain preservation: retaining the full interior of the physical safe set does not guarantee moderate control demand. In contrast, auxiliary-domain constructions may exclude some such states from their certified regions rather than certify them under potentially large corrective inputs.

\begin{table}[t]
\caption{Minimum safety margins and maximum requested inputs for the mobile robot simulation cases.}
\label{tab:nh_summary}
\centering
\renewcommand{\arraystretch}{1.1}
\begin{tabular}{c c c c c}
\hline
Case & $\min d(t)$ & $\min h_0(t)$ & $\max |u_1(t)|$ & $\max |u_2(t)|$ \\
\hline
1 & $0.3055$ & $0.7044$ & $68.22$ & $30.97$ \\
2 & $0.2991$ & $0.6877$ & $15.62$ & $37.33$ \\
3 & $0.2064$ & $0.4553$ & $30.37$ & $17.03$ \\
4 & $0.4977$ & $1.2432$ & $149.24$ & $91.24$ \\
5 & $0.5039$ & $1.2617$ & $55.91$ & $18.18$ \\
\hline
\end{tabular}
\end{table}

%%----------------------------------------------------------------
\section{Conclusion} \label{sec:conclusion}

This paper proposed a scaling-based reciprocal CBF construction for force-controlled nonholonomic mobile robots subject to relative-degree-two safety constraints. By placing a positive motion-dependent scaling factor in the numerator of a reciprocal barrier, the proposed construction preserves the certified interior domain of the original physical safe set while restoring first-order control authority. For the considered robot model, we established the reciprocal-CBF property and proved forward invariance of the physical safe-set interior under controllers satisfying the induced reciprocal-CBF condition. A scalar strict-feedback system was further used to provide a structural interpretation of the kinematic--dynamic cascade underlying the proposed construction, without claiming a general solution for arbitrary higher-relative-degree or nonholonomic systems.
A fully systematic or optimal synthesis and tuning procedure for scaling functions in broader nonholonomic models is left for future work.
Numerical examples demonstrated the induced safe-set geometry and its integration with an optimization-based obstacle-avoidance controller. These results suggest that scaling-based reciprocal barriers offer a promising way to avoid additional certified-domain shrinkage while retaining motion-dependent safety shaping.

%% -------------------------------------------------------------------
\section*{Acknowledgment} %% ASME requests this exact spelling, singular.
Bo Wang would like to express his gratitude to Professor Miroslav Krsti{\'c} of UCSD for the fruitful discussions regarding the construction of the control barrier functions. The authors also thank the Associate Editor and the anonymous reviewers for their constructive comments and suggestions.

The work of Bo Wang was partially supported by the PSC-CUNY Research Award from The City University of New York.

\section*{Conflict of Interest}
There are no conflicts of interest.

\section*{Data Availability Statement}
The datasets generated and supporting the findings of this article are obtainable from the corresponding author upon reasonable request.

%% -------------------------------------------------------------------
\bibliographystyle{asmejour}   
%% .bst file that follows ASME journal format. Do not change.
\bibliography{asmejour} %% <=== change this to name of your bib file
%% To omit final list of figures and tables, use the class option [nolists]
\end{document}